\documentclass[10pt, conference]{IEEEtran}

\usepackage{amsmath,times,amssymb,amsthm,epsfig,nicefrac,color}
\usepackage{amsfonts,graphics,epsfig,graphics,caption,subcaption}
\usepackage{verbatim,makeidx}
\usepackage[mathcal]{euscript} 
\usepackage{url}
\usepackage{dsfont}
\usepackage{cite,calc}
\usepackage{algorithm}
\usepackage{algorithmic}
\usepackage{pbox}
\usepackage{bbm}
\newcounter{ALC@tempcntr}
\theoremstyle{plain}
\newtheorem{theorem}{Theorem}
\newtheorem{proposition}{Proposition}

\theoremstyle{definition}

\theoremstyle{remark}
\newtheorem{remark}{Remark}

\newcommand{\beq}{\begin{eqnarray}}
\newcommand{\eeq}{\end{eqnarray}}

\newcommand{\field}[1]{\mathbb{#1}}

\newcommand{\F}{\field{F}}

\newfont{\bbb}{msbm10 scaled 500}

\newfont{\bb}{msbm10 scaled 1100}

\newcommand{\FF}{\mbox{\bb F}}

\newcommand{\cv}{{\bf c}}

\newcommand{\xv}{{\bf x}}
\newcommand{\yv}{{\bf y}}

\newcommand{\Em}{{\bf E}}

\newcommand{\Hm}{{\bf H}}
\newcommand{\Id}{{\bf I}}
\newcommand{\Jm}{{\bf J}}

\newcommand{\Cc}{{\cal C}}

\newcommand{\Mc}{{\cal M}}

\newcommand{\Sc}{{\cal S}}

\newcommand{\remove}[1]{}

\definecolor{OXO-emph}{RGB}{153,0,0}

\DeclareMathAlphabet{\mathpzc}{OT1}{pzc}{m}{it}

\theoremstyle{definition}

\theoremstyle{remark}

\newcommand{\latexe}{{\LaTeX\kern.125em2%
                      \lower.5ex\hbox{$\varepsilon$}}}

\chardef\bslash=`\\	

\makeatletter		
\def\square{\RIfM@\bgroup\else$\bgroup\aftergroup$\fi
\vcenter{\hrule\hbox{\vrule\@height.6em\kern.6em\vrule}
\hrule}\egroup}\makeatother\makeindex

\IEEEoverridecommandlockouts

\begin{document}
\sloppy

\title{Progress on High-rate MSR Codes: Enabling Arbitrary Number of Helper Nodes}

 \author{
\IEEEauthorblockN{Ankit Singh Rawat}
\IEEEauthorblockA{CS Department\\
     Carnegie Mellon University\\
     Pittsburgh, PA 15213\\
     Email:~asrawat@andrew.cmu.edu.}
\and
\IEEEauthorblockN{O. Ozan Koyluoglu}
   \IEEEauthorblockA{Department of ECE\\
     The University of Arizona\\
     Tucson, AZ 85721\\
     Email:~ozan@email.arizona.edu.}
\and
\IEEEauthorblockN{Sriram Vishwanath}
   \IEEEauthorblockA{Department of ECE\\
     The University of Texas at Austin\\
     Austin, TX 78712\\
     Email:~sriram@austin.utexas.edu.}
 }

\interdisplaylinepenalty=2500 

\allowdisplaybreaks

\maketitle



\begin{abstract}
This paper presents a construction for high-rate MDS codes that enable bandwidth-efficient repair of a single node. Such MDS codes are also referred to as the minimum storage regenerating (MSR) codes in the distributed storage literature. The construction presented in this paper generates MSR codes for all possible number of helper nodes $d$ as $d$ is a design parameter in the construction. Furthermore, the obtained MSR codes have polynomial sub-packetization (a.k.a. node size) $\alpha$. The construction is built on the recent code proposed by Sasidharan et al.~\cite{SAK15}, which works only for $d = n-1$, i.e., where all the remaining nodes serve as the helper nodes for the bandwidth-efficient repair of a single node. The results of this paper broaden the set of parameters
where the constructions of MSR codes were known earlier.
\end{abstract}

\begin{IEEEkeywords}
Codes for distributed storage, regenerating codes, minimum storage regenerating (MSR) codes, sub-packetization.
\end{IEEEkeywords}



\section{Introduction}

Consider a distributed storage system with $n$ storage nodes which stores a file of size $\Mc$ symbols over a finite field. The distributed storage system (DSS) is referred to be an $(n, k)$-DSS if it has {\em `any $k$ out of $n$' property}, i.e., the content of any $k$ out of $n$ storage nodes is sufficient to reconstruct the entire file. In \cite{dimakis}, Dimakis et al. explore the issue of node repair in an $(n, k)$-DSS. In particular, they study $(n, k)$-DSS which allow for the repair of a single failed node by contacting $d$ out of $n-1$ remaining storage nodes and downloading $\beta$ symbols from each of these $d$ helper nodes. Assuming that each node in the system stores $\alpha$ symbols (over the finite field), Dimakis et al. obtain a trade-off between the node size $\alpha$ and repair bandwidth $\gamma = d\beta$, the amount of data downloaded during the repair process. The codes that attain this trade-off are referred to as {\em regenerating codes}. The two extreme points of this trade-off correspond to the minimum possible storage and the minimum possible repair-bandwidth for an $(n, k)$-DSS. These two points are termed as {\em minimum storage regenerating (MSR)} point and {\em minimum bandwidth regenerating (MBR)} point, respectively. The MSR point corresponds to
\begin{align}
\left(\alpha_{\rm MSR}, \beta_{\rm MSR}\right) &= \left(\frac{\Mc}{k}, \frac{d}{d - k + 1}\frac{\Mc}{k}\right). \nonumber
\end{align}
The MBR point is defined by
\begin{align}
\left(\alpha_{\rm MBR}, \beta_{\rm MBR}\right) &= \left(\frac{2d}{2d - k + 1}\frac{\Mc}{k}, \frac{2}{(2d - k + 1)}\frac{\Mc}{k}\right). \nonumber
\end{align}
The codes achieving the MSR and the MBR points are referred to as {\em minimum storage regenerating (MSR) codes} and {\em minimum bandwidth regenerating (MBR) codes}, respectively. Note that the MSR codes are also maximum-distance separable (MDS) codes~\cite{MacSlo}.

In \cite{dimakis}, Dimakis et al. also show the existence of the codes that achieve every point on the $\alpha$ vs. $d\beta$ trade-off for all possible system parameters $n, k, d$ to ensure {\em functional repair}. Under the functional repair, the content of the repaired node may differ from that of the failed node. However, the repaired node does ensure the `any $k$ out of $n$' property of the system. Sometimes, due to various system level requirements, it is desirable to construct regenerating codes that ensure {\em exact repair} of the failed node, i.e., the content of the repaired node is the same as the content of the failed node. In \cite{RSK11}, Rashmi et al. settle the problem of designing exact repairable MBR codes ({\em exact-MBR codes})  as they propose an explicit construction of such codes for all possible system parameters $n, k$ and $d$. 

On the other hand, the problem of constructing the exact-MSR codes has not been fully understood yet. The exact-MSR codes with $k < 3$ and $k \leq \frac{n}{2}$ are presented in \cite{WuDim09} and \cite{SRKR_itw10, SuhRam_isit10}, respectively. In \cite{RSK11}, Rashmi et al. present explicit constructions for exact-MSR codes with $2k - 2 \leq d \leq  n - 1$. In general, all of these constructions correspond to exact-MSR codes of low rate with $\frac{k}{n} \leq \frac{1}{2} + \frac{1}{2n}$. In \cite{CJMRS13}, Cadambe et al. show the existence of high-rate exact MSR codes when node size $\alpha$ (also referred to as {\em sub-packetization level}) approaches to infinity. Towards constructing high-rate exact-MSR codes with finite sub-packetization level, Papailiopoulos et al. utilize Hadamard matrices to construct exact-MSR codes with $n - k = 2$ and $d = n -1$ in \cite{PapDimCad_hadamard}. Using permutation-matrices exact-MSR codes for all $(n, k)$ pairs with $d = n - 1$ which only ensure repair bandwidth-efficient repair of systematic nodes are presented in \cite{cadambe2011optimal} and \cite{zigzag13}. In \cite{zigzag_allerton11}, Wang et al. generalize these constructions to enable repair of all nodes with $d = n-1$ helper nodes. However, we note that the sub-packetization level $\alpha$ of the constructions presented in \cite{zigzag13, zigzag_allerton11, PapDimCad_hadamard, cadambe2011optimal} is exponential in $k$. 

Recently, Sasidharan et al. have presented a construction of a constant  (high) rate MSR codes with polynomial sub-packetization in \cite{SAK15}. This construction enables repair of all the nodes in the system and works for $d = n-1$, i.e., all the remaining $n-1$ nodes has to be contacted to repair a single failed node. The construction with polynomial sub-packetization and enabling repair of only systematic nodes are also presented in \cite{WTB12, Cadambe_poly}. As for the converse results, Goparaju~et al. establish a lower bound on the sub-packetization level of an MSR code with given $n$ and $k$ in \cite{GTC14}.

In this paper, we present a construction for exact-MSR codes that allow for any given number of helper nodes, i.e., $k \leq d \leq n -1$.  In addition to working for an arbitrary (but fixed) $d$, our construction possesses the desirable properties of having polynomial sub-packetization level for a constant rate and enabling repair-bandwidth efficient repair of all the nodes in the system. We obtain this construction by suitably modifying the construction of Sasidharan et al.~\cite{SAK15}. The rest of the paper is organized as follows. We introduce the notation and necessary background in Section~\ref{sec:prelims}. In Section~\ref{sec:MSR_construction}, we present our code construction. In Section~\ref{sec:repair}, we describe the node repair process for the proposed code construction. We establish the MDS property (a.k.a. `any k out of n' property) for the construction in Section~\ref{sec:MDSprop}. We conclude the paper in Section~\ref{sec:conclusion}.

\section{Preliminaries}
\label{sec:prelims}

Let $\mathbbm{1}_{\{\cdot\}}$ denote the standard indicator function which takes the value $1$ if the condition stated in $\{\cdot\}$ is true and takes the value $0$ otherwise. For two $n\alpha$-length vectors $\xv$ and $\yv$, we defined the Hamming distance between them as follows. 
$$d_{\rm H}(\xv, \yv) = \sum_{i = 1}^{n}\mathbbm{1}_{\{\xv_i \neq \yv_i \}},$$
where for $i \in [n]$, we have $\xv_i = (x_{(i-1)\alpha + 1},\ldots,x_{i\alpha})$ and $\yv_i = (y_{(i-1)\alpha + 1},\ldots,y_{i\alpha})$. We say that a set of vectors $\Cc \subseteq \F^{n\alpha}_{Q}$ is an $(n, M, d_{\min}, \alpha)_Q$ vector code if we have $|\Cc| = M$ and $d_{\min} = \min_{\xv, \yv \in \Cc}d_{\rm H}(\xv, \yv)$.  Given a codeword $\cv = (c_1, c_2,\ldots, c_{n\alpha}) \in \Cc$, we use $\cv_i = (c_{(i - 1)\alpha + 1}, c_{(i - 1)\alpha + 2},\ldots, c_{i\alpha})$ to denote the $i$-th vector (code) symbol in the codeword. When the code $\Cc$ spans a linear subspace of dimension $\log_Q{M}$, we call $\Cc$ to be a linear vector code and refer to it as an $[n, \log_Q{M}, d_{\min}, \alpha]_{Q}$ vector code. Note that an $[n, k\alpha, d_{\min}, \alpha]_Q$ vector code can be defined by a parity-check matrix
$$
\Hm = \left( \begin{array}{cccc}
H_{1,1} & H_{1, 2} & \cdots & H_{1,n}\\
H_{2,1} & H_{2, 2} & \cdots & H_{2,n}\\
\vdots & \vdots & \ddots & \vdots \\
H_{n-k,1} & H_{n-k, 2} & \cdots & H_{n-k,n}\\
\end{array} \right) \in \F_Q^{(n - k)\alpha \times n\alpha},
$$
where each $H_{i, j}$ is an $\alpha \times \alpha$ matrix with its entries belonging to $\F_Q$. For a set $\Sc = \{i_1, i_2,\ldots, i_{|\Sc|}\} \subseteq [n]$, we define the $(n - k)\alpha \times |\Sc|\alpha$ matrix $\Hm(:,\Sc)$ as follows. 
$$
\Hm(:,\Sc) = \left( \begin{array}{cccc}
H_{1, i_1} & H_{1, i_2} & \cdots & H_{1, i_{|\Sc|}}\\
H_{2, i_1} & H_{2, i_2} & \cdots & H_{2, i_{|\Sc|}}\\
\vdots & \vdots & \ddots & \vdots \\
H_{n - k, i_1} & H_{n - k, i_2} & \cdots & H_{n- k, i_{|\Sc|}}
\end{array} \right).
$$
Note that the matrix $\Hm(:, \Sc)$ comprises those coefficients in the linear constraints defined by the parity-check matrix $\Hm$ that are associated with the vector code symbols indexed by the set $\Sc$.

\section{Code Construction}
\label{sec:MSR_construction}

In what follows, we use $\Sigma$ to represent a linear combination whose coefficients are not specified explicitly. For example, for $a_1, a_2,\ldots, a_r \in \FF_Q$, $\sum_{i = 1}^ra_i$ denotes a linear combination of these $r$ elements where unspecified coefficients of the linear combination belong to $\FF_Q$. For an integer $q > 0$, we use $[q]$ and $[0:q-1]$ to denote the sets $\{1, 2,\ldots, q\}$ and $\{0, 1,\ldots, q-1\}$, respectively.

Assume that $n = (t - 1)(d - k + 1) + s$, for $t > 1$ and $0 \leq s \leq d-k$. We take 
\begin{align}
\alpha = \begin{cases} (d - k + 1)^{t-1} = q^{t-1} & \mbox{if } s = 0  \\
(d - k + 1)^t = q^t & \mbox{otherwise}.
\end{cases}
\end{align}
Note that we use $q$ to denote $d - k + 1$. Moreover, as compared to \cite{SAK15}, we describe the construction for the wider range of parameters which corresponds to $s >  0$. Therefore, for $s>0$, we have $\alpha = (d - k + 1)^t = q^t$. For these values of parameters, at the MSR point, a node repair step involves downloading $$\beta = \frac{\alpha}{d - k + 1} = (d - k + 1)^{t-1} =q^{t-1}$$ symbols from each of the $d$ contacted nodes. Let $n = (t-1)q + s$ nodes be indexed by tuples 
\begin{align}
\label{eq:node_set}
\mathcal{N} &= \left\{(i, \theta):~(i, \theta) \in [t-1]\times [0:q-1]\right\}~\cup \nonumber \\
&~~~~~\left\{(t, \theta):~\theta \in [0:s-1]\right\}.
\end{align}
Note that each node in the system stores $\alpha = q^t$ code symbols. Let  $\{c((x_1, x_2,\ldots, x_t); (i, \theta)) \}_{(x_1,\ldots, x_t) \in [0:q-1]^t}$ represent the $q^t$ code symbols stored on the $(i, \theta)$-th node. In order to specify the MSR code $\Cc$, we specify $(n-k)\alpha = (n-k)q^t$ linear constraints over $\FF_Q$ that each codeword in $\Cc$ has to satisfy. We partition these $(n-k)\alpha$ constraints into two types of constraints which we refer to as ${\rm Type~I}$ and ${\rm Type~II}$ constraints, respectively.

\textbf{{\rm Type~I} constraints:} For each $(x_1,\ldots, x_t) \in [0:q-1]^t$, we have $n - d$ constraints of the following form.
\begin{align}
\label{eq:type1}
&\sum_{\theta \in [0:q-1]}c((x_1,\ldots, x_t); (1, \theta)) +  \nonumber \\
&\sum_{\theta \in [0:q-1]}c((x_1,\ldots, x_t); (2, \theta))  + \cdots + \nonumber \\
&\sum_{\theta \in [0:q-1]}c((x_1,\ldots, x_t); (t-1, \theta)) + \nonumber \\
&\sum_{\theta \in [0:s-1]}c((x_1,\ldots, x_t); (t, \theta))   = 0.
\end{align}
The coefficients of these constraints are chosen in such a way that the following holds for each $(x_1, x_2,\ldots, x_t) \in [0:q-1]^t$. Given any subset of $d$ code symbols out of $n$ code symbols  $\{c((x_1, x_2,\ldots, x_t);(i,\theta)\}_{(i, \theta)) \in \mathcal{N}}$, the remaining  $n-d$ code symbols can be recovered using these ${\rm Type~I}$ constraints.

\textbf{{\rm Type~II} constraints:} We now described the remaining $(n - k)\alpha - (n - d)\alpha = (d - k)\alpha = (d - k)q^{t}$ constraints satisfied by the codewords. For every $(x_1, x_2,\ldots, x_t) \in [0:q-1]^t$ and $\Delta \in [1:q-1]$, we have
\begin{align}
\label{eq:type2}
&c((x_1 - \Delta, x_2,\ldots, x_t); (1, x_1)) + \nonumber \\
&c((x_1, x_2 - \Delta,\ldots, x_t); (2, x_2)) + \cdots + \nonumber \\
&c((x_1,\ldots, x_{t-1} - \Delta, x_t); (t-1, x_{t-1})) + \nonumber \\
&\underline{c((x_1, x_2,\ldots, x_t - \Delta); (t, x_t))} +  \nonumber \\
&\sum_{\theta \in [0:q-1]}c((x_1,\ldots, x_t); (1, \theta)) + \nonumber \\
& \sum_{\theta \in [0:q-1]}c((x_1,\ldots, x_t); (2, \theta)) + \cdots + \nonumber \\
&\sum_{\theta \in [0:q-1]}c((x_1,\ldots, x_t); (t-1, \theta)) + \nonumber \\
& \sum_{\theta \in [0:s-1]}c((x_1,\ldots, x_t); (t, \theta))  = 0.
\end{align}
Here, the computation $x_i  - \Delta$, for $i \in [t]$, is performed modulo $q$. Furthermore, the underlined code symbols $c((x_1, x_2,\ldots, x_t - \Delta); (t, x_t))$ correspond to $0$ for $x_t \geq s$ as there is no node which is indexed by the tuple $(t, x_t)$ with $x_t \geq s$.

\begin{remark}
\label{rem:diff1}
One key difference from the construction in \cite{SAK15} is that for each tuple $(x_1,\ldots, x_t) \in [0:q-1]^t$, we generate $n - d$ {\rm Type~I} constraints. In \cite{SAK15}, only $1$ such constraint was generated as the case of $d  = n - 1$ was considered. The coefficients of these constraints need to be carefully chosen to ensure the requirements specified after \eqref{eq:type1}. We address this issue in Remark~\ref{rem:diff2}.
\end{remark}

\section{Recovering a failed node}
\label{sec:repair}

Assume that the node indexed by the tuple $(i, \theta_0)$ fails. We now describe the repair process of the failed node. The repair process can be viewed to have two stages. In the first stage, we use the {\rm Type~I} constraints to recover $\beta = \frac{\alpha}{d - k + 1} = q^{t - 1}$ out of $\alpha = q^t$ code symbols that are lost due to the node failure. Towards this, from each of the $d$ contacted nodes, we download the code symbols indexed by the tuples $\{(x_1,\ldots,x_{i-1}, \theta_0, x_{i+1},\ldots, x_t)\}$, where $(x_1,\ldots,x_{i-1}, x_{i+1},\ldots, x_t)$ span over all values in  $[0:q-1]^{t-1}$ from each of the $d$ contacted nodes. These symbols along with the {\rm Type~I} constraints (cf.~\eqref{eq:type1}) allow us to recover the symbols 
\begin{align}
\label{eq:stage1}
c((x_1,\ldots,x_{i-1}, \theta_0, x_{i+1},\ldots, x_t);(i,\theta)),
\end{align}
for every $(x_1,\ldots,x_{i-1}, x_{i+1},\ldots, x_t) \in [0:q-1]^{t-1}$ and $(i,\theta) \in \mathcal{N}$. 

It is clear from \eqref{eq:stage1} that after the first stage we have access to $\beta = q^{t-1}$ code symbols stored on the failed node as well as the $\beta$ symbols stored on the remaining $n-1$ nodes. In the second stage, we employ the {\rm Type~II} constraints (cf.~\eqref{eq:type2}) to recover the remaining $(d- k)\beta = (q-1)q^{t-1}$ symbols stored on the failed node, i.e., the node indexed by the tuple $(i, \theta_0)$. Recall that for a tuple $(x_1,\ldots, x_{i-1}, \theta_0, x_{i+1},\ldots, x_t) \in [0:q-1]^t$ and a non-zero integer $\Delta \in [q-1]$, the corresponding {\rm Type~II} constraint is as follows:
\begin{align}
\label{eq:stage2}
&c((x_1 - \Delta,\ldots,x_{i-1}, \theta_0, x_{i+1},\ldots, x_t); (1, x_1))+ \cdots + \nonumber \\
&c((x_1,\ldots, x_{i-1} - \Delta, \theta_0, x_{i+1},\ldots x_t); (i-1, x_{i-1})) + \nonumber \\
& \underline{c((x_1,\ldots, x_{i-1}, \theta_0 - \Delta, x_{i+1},\ldots, x_t); (i, \theta_0))} + \nonumber \\
&c((x_1,\ldots, x_{i-1},\theta_0, x_{i+1} - \Delta,\ldots, x_t); (i+1, x_{i+1})) + \cdots + \nonumber \\
&c((x_1,\ldots, x_{i-1}, \theta_0, x_{i+1},\ldots, x_t -\Delta); (t, x_t)) + \nonumber \\
& \Big( \sum_{\theta \in [0:q-1]}c((x_1,\ldots, x_{i-1}, \theta_0, x_{i+1},\ldots, x_t); (1, \theta)) + \nonumber \\
&\sum_{\theta \in [0:q-1]}c((x_1,\ldots, x_{i-1}, \theta_0, x_{i+1},\ldots, x_t); (2, \theta))  +\cdots +  \nonumber \\
&\sum_{\theta \in [0:q-1]}c((x_1,\ldots, x_{i-1}, \theta_0, x_{i+1},\ldots, x_t); (t, \theta)) \Big)= 0.
\end{align}
Note that except the underlined code symbol we know every other code symbol involved in \eqref{eq:stage2} (cf.~\eqref{eq:stage1}). Therefore, using the constraints in \eqref{eq:stage2}, we can complete the second stage of the repair process which recovers the remaining $(q-1)q^{t-1}$ code symbols from the failed node.

\section{MDS property of the code}
\label{sec:MDSprop}

In this section, we prove that it is possible to obtain the codes from the construction described in Section~\ref{sec:MSR_construction} that are maximum-distance separable (MDS). In particular, we argue that if the coding coefficients in the construction are selected from a finite field of large enough size, then there exists a choice for coding coefficients which lead to the obtained code being an MDS code. (We note that the argument presented in this section follows very closely to the argument used in \cite{SAK15}.)

Recall that for 
a code $\Cc$ defined in Section~\ref{sec:MSR_construction}, we can represent a codeword in the code $\Cc$ by an $n\alpha$-length vector in $\F_Q^{n\alpha}$. In particular, let $\cv = (\cv_1, \cv_2,\ldots, \cv_n)$ be a generic codeword from the code $\Cc$. Here, for each $j \in [n]$, $\cv_j \in \F_Q^{\alpha}$ represent the code symbols stored on the $i$-th node in the system. Assuming that the node indexed by the tuple $(i, \theta)$ represents the $\left((i - 1)q + \theta + 1\right)$-th node in the system, we have $\cv_{(i - 1)q + \theta + 1} = \{c((x_1, x_2,\ldots, x_t); (i, \theta)) \}_{(x_1,\ldots, x_t) \in [0:q-1]^t}$.

Let $\Hm \in \F_Q^{(n - k)\alpha \times n\alpha}$ be the parity check matrix of the code $\Cc$ defined by the {\rm Type~I} and {\rm Type~II} linear constraints presented in \eqref{eq:type1} and \eqref{eq:type2}, respectively. Note that it follows from the code construction that the parity check matrix $\Hm$ has the following structure. 
\begin{align}
\Hm = \left(\begin{array}{c}
\Hm^{\rm I} \\ \hline
\Hm^{\rm II} \end{array} \right).
\end{align}
Here, $\Hm^{\rm I}$ is an $(n - d)\alpha \times n\alpha$ matrix over $\F_Q$ which is defined by the $(n - d)\alpha = (n - d)q^t$ ${\rm Type~I}$ constraints (cf.~\eqref{eq:type1}). On the other hand, the $(d - k)\alpha = (d - k)q^t$ ${\rm Type~II}$ constraints (cf.~\eqref{eq:type2}) constitute the $(d - k)\alpha \times n\alpha$ matrix $\Hm^{\rm II}$ over $\F_Q$. We now focus on the structure of the two matrices $\Hm^{\rm I}$ and $\Hm^{\rm II}$. Note that 
\begin{align}
\Hm^{\rm I}  = \left(\begin{array}{c}
H^{\rm I}_1 \\
H^{\rm I}_2 \\
\vdots \\
H^{\rm I}_{n - d}
\end{array}\right)
\end{align}
where, for $i \in [n - d]$, the matrix $H^{\rm I}_i \in \F_Q^{\alpha \times n\alpha}$ is obtained by taking one of the $n - d$ ${\rm Type~I}$ constraints associated with each of the $q^t$ values for the tuple $(x_1, x_2,\ldots, x_t) \in [0: q-1]^t$ (cf.~\eqref{eq:type1}).  Similarly, we have 
\begin{align}
\Hm^{\rm II}  = \left(\begin{array}{c}
H^{\rm II}_1 \\
H^{\rm II}_2 \\
\vdots \\
H^{\rm II}_{d - k}
\end{array}\right),
\end{align}
where, for $i \in [d - k]$, the matrix $H^{\rm II}_i \in \F_Q^{\alpha \times n\alpha}$ is defined by one of the $d - k$ ${\rm~Type~II}$ constraints corresponding to each of the $q^t$ values for the tuple $(x_1, x_2,\ldots, x_t) \in [0: q-1]^t$ (cf.~\eqref{eq:type2}). Recall that for a tuple $(x_1, x_2,\ldots, x_t) \in [0: q-1]^t$, the $d - k$ ${\rm~Type~II}$ constraints corresponding to the tuple are associated with the $d - k$ values of the parameter $\Delta \in [1:q-1]$ (cf.~\eqref{eq:type2}). Exploring the structure of the parity check matrix further, we note that for every $i \in [n - d]$, the $\alpha \times n\alpha$ matrix $H^{\rm I}_i$ is a block matrix consisting of $n$ blocks where each blocks is an $\alpha \times \alpha$ diagonal matrix over $\F_Q$. In particular, let's denote it as 
\begin{align}
H^{\rm I}_i = \left( \begin{array}{c|c|c|c}J^{\rm I}_i(1) & J^{\rm I}_i(2)& \ldots & J^{\rm I}_i(n) \end{array} \right),
\end{align}
where $J^{\rm I}_i(j) \in \F_Q^{\alpha \times \alpha}$ is a diagonal matrix with all of its diagonal entries being non-zero. On the other hand, for $i \in [d - k]$, the $\alpha \times n\alpha$ matrix $H^{\rm II}_i$ is also a block matrix which can be written in the following form. \begin{align}
H^{\rm II}_i =  \left( \begin{array}{c|c|c|c}H^{\rm II}_i(1) & H^{\rm II}_i(2)& \ldots & H^{\rm II}_i(n) \end{array} \right),
\end{align}
where $H^{\rm II}_i(j) = J^{\rm II}_i(j) + E^{\rm II}_i(j) \in \F_Q^{\alpha \times \alpha}$. In this sum, the matrix $J^{\rm II}_i(j) \in \F_Q^{\alpha \times \alpha}$ is a diagonal matrix with all of its diagonal entries being non-zero. On the other hand, the second matrix in the sum $E^{\rm II}_i(j) \in \F_Q^{\alpha \times \alpha}$ has at most $1$ non-zero element in each of its row. In particular, for every $i \in [d-k]$, the block matrix
\begin{align}
\left( \begin{array}{c|c|c|c}E^{\rm II}_i(1) & E^{\rm II}_i(2)& \ldots & E^{\rm II}_i(n) \end{array} \right)
\end{align}
has exactly $t$ non-zero elements in each of its row. Here, we note that the matrix $E^{\rm II}_i$ contains coefficients of the following part of those $\alpha = q^t$ {\rm Type~II} constraints which correspond to a fixed value of the parameter $\Delta \in [d - k]$ (cf.~\eqref{eq:type2}). 
\begin{align}
\label{eq:type2_Epart}
&c((x_1 - \Delta, x_2,\ldots, x_t); (1, x_1)) + \cdots + \nonumber \\
&c((x_1,\ldots, x_{t-1} - \Delta, x_t); (t-1, x_{t-1})) + \nonumber \\ 
&c((x_1, x_2,\ldots, x_t - \Delta); (t, x_t)).
\end{align}
With all the components of the parity check matrix defined, we can represent the parity check matrix $\Hm$ as sum of two $(n - k)\alpha \times n\alpha$ matrix as follows. 
\begin{align}
\label{eq:HJE}
\Hm =  \left(\begin{array}{c}
\Hm^{\rm I} \\ \hline
\Hm^{\rm II} \end{array} \right) =  \Jm + \Em.
\end{align}
Here $\Jm \in \F_Q^{(n - k)\alpha \times n\alpha}$ and $\Em \in \F_Q^{(n - k)\alpha \times n\alpha}$ denote the following matrices.
\begin{align}
\label{eq:Jdef}
\Jm = \left(\begin{array}{c|c|c|c}
J^{\rm I}_1(1) & J^{\rm I}_1(2)& \cdots& J^{\rm I}_1(n) \\
J^{\rm I}_2(1) & J^{\rm I}_2(2)&  \cdots & J^{\rm I}_2(n) \\
\vdots & \vdots &  \ddots & \vdots \\
J^{\rm I}_{n - d}(1) & J^{\rm I}_{n - d}(2)& \cdots & J^{\rm I}_{n - d}(n) \\
\hline
J^{\rm II}_1(1) & J^{\rm II}_1(2)&  \cdots & J^{\rm II}_1(n) \\
J^{\rm II}_2(1) & J^{\rm II}_2(2)&  \cdots & J^{\rm II}_2(n) \\
\vdots & \vdots &  \ddots & \vdots \\
J^{\rm II}_{d - k}(1) & J^{\rm II}_{d - k}(2)& \cdots & J^{\rm II}_{d - k}(n) \\
\end{array} \right), \\
\Em = \left(\begin{array}{c|c|c|c}
\mathbf{0}_{\alpha} & \mathbf{0}_{\alpha} &\cdots& \mathbf{0}_{\alpha} \\
\mathbf{0}_{\alpha} & \mathbf{0}_{\alpha} & \cdots & \mathbf{0}_{\alpha} \\
\vdots & \vdots & \ddots & \vdots \\
\mathbf{0}_{\alpha} & \mathbf{0}_{\alpha} & \cdots & \mathbf{0}_{\alpha} \\
\hline
E^{\rm II}_1(1) & E^{\rm II}_1(2)&   \cdots & E^{\rm II}_1(n) \\
E^{\rm II}_2(1) & E^{\rm II}_2(2)&  \cdots & E^{\rm II}_2(n) \\
\vdots & \vdots & \ddots & \vdots \\
E^{\rm II}_{d - k}(1) & E^{\rm II}_{d - k}(2) &  \cdots & E^{\rm II}_{d - k}(n) \\
\end{array} \right).\label{eq:Edef} 
\end{align}
Note that, we use $\mathbf{0}_{\alpha}$ to represent the $\alpha \times \alpha$ all zero matrix. We now specify the non-zero entries in both the matrices $\Jm$ and $\Em$. Let $H_{MDS}$ be an $(n-k) \times n$ Cauchy matrix, 
\begin{align}
H_{MDS} = \left( \begin{array}{cccc} 
\frac{1}{a_1 - b_1}&\frac{1}{a_1 - b_2}&\cdots& \frac{1}{a_1 - b_n} \\
\frac{1}{a_2 - b_1}&\frac{1}{a_2 - b_2}&\cdots& \frac{1}{a_2 - b_n} \\
\vdots & \vdots & \ddots & \vdots \\
\frac{1}{a_{n - k} - b_1} & \frac{1}{a_{n-k} - b_2} &\cdots& \frac{1}{a_{n - k} - b_n} \end{array} \right), 
\end{align}
where $\{a_1, a_2,\ldots, a_{n - k}, b_1, b_2,\ldots, b_n \}$ are $2n - k$ distinct elements from the field $\F_Q$. Assuming that $\mathbf{I}_{\alpha}$ denotes the $\alpha \times \alpha$ identity matrix, we define the matrix $\Jm$ (cf.~\eqref{eq:Jdef}) as follows. 
\begin{align}
\Jm = H_{MDS} \otimes \mathbf{I}_{\alpha},
\end{align}
where $\otimes$ denotes the Kronecker product between two matrices. As for non-zero elements in the matrix $\Em$, we set all of its non-zero elements to be an indeterminate $\rho \in \F^{\ast}_Q$. In order to make it more clear, we denote the obtained matrix as $\Em^{\rho}$ and accordingly the parity-check matrix defined in \eqref{eq:HJE} becomes
\begin{align}
\label{eq:HJM1}
\Hm = \Jm + \Em^{\rho} = H_{MDS}\otimes \Id_{\alpha} + \Em^{\rho}.
\end{align} 
Next, we show that for large enough $Q$, there exists a choice for $\rho$ which makes the code defined by the parity check matrix $\Hm$ an MDS code. However, before showing this, we argue that our choice of the matrix $\Jm$ meets the requirement for the ${\rm Type~I}$ constraints. This requirement states that for every $(x_1,\ldots, x_t) \in [0:q-1]^t$, given any subset of $d$ code symbols out of $n$ code symbols  $\{c((x_1, x_2,\ldots, x_t);(i,\theta))\}_{(i, \theta) \in \mathcal{N}}$, the remaining  $n-d$ code symbols can be recovered using the corresponding ${\rm Type~I}$ constraints (cf.~\eqref{eq:type1}). This requirement indeed holds as for a tuple $(x_1,\ldots, x_t) \in [0:q-1]^t$, the coefficients associated with its ${\rm Type~I}$ constraints are the elements of the following $(n - d) \times n$ sub-matrix of $H_{MDS}$. 
\begin{align}
H^{\rm I}_{MDS} = \left( \begin{array}{cccc} 
\frac{1}{a_{1} - b_1}&\frac{1}{a_{1} - b_2}&\cdots& \frac{1}{a_{1} - b_n} \\
\frac{1}{a_{2} - b_1}&\frac{1}{a_{2} - b_2}&\cdots& \frac{1}{a_{2} - b_n} \\
\vdots & \vdots &\ddots& \vdots \\
\frac{1}{a_{{n - d}} - b_1} &\frac{1}{a_{{n-d}} - b_2}&\cdots& \frac{1}{a_{{n - d}} - b_n} \end{array} \right). \nonumber 
\end{align}
Since any $(n - d) \times (n-d)$ sub-matrix of $H^{\rm I}_{MDS}$ is full-rank, given any subset of $d$ code symbols out of $n$ code symbols  $\{c((x_1, x_2,\ldots, x_t);(i,\theta))\}_{(i, \theta) \in \mathcal{N}}$, the remaining  $n-d$ code symbols can indeed be recovered.

\begin{remark}
\label{rem:diff2}
Note that each tuple $(x_1,\ldots, x_t) \in [0:q-1]^t$ has only $1$ associated {\rm Type~I} constraint in \cite{SAK15}. Therefore the requirement on  $H^{\rm I}_{MDS}$ reduces to having all of its elements non-zero. On the other hand for $d \neq n - 1$ case, we have an additional requirement that any $(n - d) \times (n-d)$ sub-matrix of $H^{\rm I}_{MDS}$ is full-rank.
\end{remark}

In order to show that for a suitable choice for the value of the indeterminate $\rho$ the code $\Cc$ defined by the matrix $\Hm$ gives an MDS code, we utilize the following standard result.

\begin{proposition}
\label{prop:mds}
Let $\Cc \in \F_{Q^{n\alpha}}$ be a linear vector code (over $\F_Q$) defined by the block parity check matrix $\Hm \in \F^{(n - k)\alpha \times n\alpha}_Q$. The code $\Cc$ is an MDS code iff for every $\Sc \subset [n]$ such that $|\Sc| = n - k$, the $(n - k)\alpha \times (n - k)\alpha$ sub-matrix $\Hm(:, \Sc)$ associated with the vector symbols indexed by the set $\Sc$ is full rank.
\end{proposition} 

\begin{theorem}
\label{thm:mds}
Let $\F_Q$ be a finite field of large enough size. Then, there exists a choice for the indeterminate $\rho$ such that the $[n, k\alpha, d_{\min}, \alpha]_Q$ vector code defined by the matrix $\Hm = H_{MDS}\otimes \Id_{\alpha} + \Em^{\rho}$ (cf.~\eqref{eq:HJM1}) is an MDS vector code, i.e., $d_{\min} = n - k + 1$.
\end{theorem}

\begin{proof}
Let $\Sc \subseteq [n]$ be a set such that $|\Sc| = n - k$. We consider the determinant of the matrix $\Hm(:, \Sc)$. Note that $\det(\Hm(:,\Sc))$ is a polynomial of the indeterminate $\rho$. Let's denote the polynomial by $f_{\Sc}(\rho)$. We have, 
\begin{align}
f_{\Sc}(\rho = 0) = \det(\Jm(:,\Sc) + \Em^{\rho = 0}(:,\Sc)) = \det(\Jm(:, \Sc)) \neq 0, \nonumber
\end{align}
where the last inequality follows as $\Jm = H_{MDS}\otimes \Id_{\alpha}$ is a parity check matrix of an MDS vector code. This establishes that $f_{\Sc}(\rho)$ is a non-trivial (not identically zero) polynomial of $\rho$. Now consider the polynomial
\begin{align}
h(\rho) = \prod_{\Sc \subseteq [n]:~|\Sc| = n-k}\det(H(:,\Sc)) = \prod_{\Sc \subseteq [n]:~|\Sc| = n-k}f_{\Sc}(\rho). \nonumber
\end{align}
Here, $h(\rho)$ is a non-trivial polynomial in $\rho$ as it is a product of non-trivial polynomials $\big\{f_{\Sc}(\rho)\big\}_{\Sc}$. Furthermore, the degree of $h(\rho)$ is bounded by ${n \choose n - k}(n - k)\alpha$. Therefore, for $Q$ large enough, there exists a value of $\rho$, say $\rho^{\ast}$ such that $h(\rho^{\ast}) \neq 0$. Combining this with Proposition~\ref{prop:mds}, we obtain that the vector code defined by the parity check matrix $\Hm = \Jm + \Em^{\rho^{\ast}}$ is an MDS vector code.
\end{proof}

\section{Conclusion}
\label{sec:conclusion}
For a given rate, we present a construction for MSR codes that allows for bandwidth-efficient repair of a single node failure with arbitrary (but fixed) number of helper nodes $d$. In addition, for the constant rate, the code has a polynomial sub-packetization (a.k.a. node size) $\alpha$. However, in the present form the construction suffers from a large field size $Q$. Note that the requirement on the field size emerges from the requirement that the code should be an MDS code (cf.~Section~\ref{sec:MDSprop}). It is an important question to resolve if the code construction with similar system parameters $n, k, d$ and polynomial sub-packetization can be achieved for a smaller field size. In particular, the lower bound on the field size for an MSR code is investigated in \cite{CM15_ITW}, and the results presented here form an upper bound.

\bibliographystyle{unsrt}
\bibliography{MSR_poly}

\end{document}